\documentclass[11pt]{article}

\usepackage{cite}

\usepackage{nicefrac}
\usepackage{amsmath}
\usepackage{amsthm}
\usepackage{amssymb}
\usepackage{mathtools}
\usepackage{framed}
\usepackage{tabularx}
\usepackage{caption}
\usepackage{lscape}
\usepackage{rotating}
\usepackage{float}
\usepackage{wrapfig}

\newcommand{\citet}[2][]{\citeauthor{#2} \cite[#1]{#2}}
\newtheorem{lemm}{Lemma}
\newtheorem{thm}{Theorem}
\newtheorem{corollary}{Corollary}

\begin{document}
\title{\vspace{-2cm} Maximum likelihood estimation for aggregated current status data: Simulation study using the illness-death model for chronic diseases with duration dependency}
\date{\today}
\author{%
	Maryam Mohammadi Saem\thanks{%
		Corresponding author: Maryam Mohammadi Saem,\quad Lehrstuhl f\"{u}r Medizinische Biometrie und Epidemiologie (MBE), Fakult\"{a}t f\"{u}r Gesundheit (Department für Humanmedizin), Universit\"{a}t Witten-Herdecke, Alfred-Herrhausen-Straße 50,
		58448 Witten, Germany; email: Maryam.MohammadiSaem@uni-wh.de%
	}\quad and\quad%
	Ralph Brinks\thanks{%
		Ralph Brinks,\quad Head of Lehrstuhl f\"{u}r Medizinische Biometrie und Epidemiologie (MBE), Fakult\"{a}t f\"{u}r Gesundheit (Department für Humanmedizin), Universit\"{a}t Witten-Herdecke, Alfred-Herrhausen-Straße 50,
		58448 Witten, Germany; email: Ralph.Brinks@uni-wh.de%
		}
}
\maketitle
\begin{abstract}
\noindent
	We use the illness-death model (IDM) for chronic conditions to derive a new analytical relation between the transition rates between the states of the IDM. The transition rates are the incidence rate ($i$) and the mortality rates of people without disease ($m_0$) and with disease ($m_1$). For the most generic case, the rates depend on age, calendar time and in case of $m_1$ also on the duration of the disease. In this work, we show that the prevalence-odds can be expressed as a convolution-like product of the incidence rate and an exponentiated linear combination of $i$, $m_0$ and $m_1$. The analytical expression can be used as the basis for a maximum likelihood estimation (MLE) and associated large sample asymptotics.
	In a simulation study according to \cite{brinks2014lexis} where a cross-sectional trial about a chronic condition is mimicked, we estimate the duration dependency of the mortality rate $m_1$ based on aggregated current status data using the ML estimator. For this, the number of study participants and the number of diseased people in eleven age groups are considered. The ML estimator provides reasonable estimates for the parameters including their large sample confidence bounds.

\end{abstract}
\\[.65em]

{\parskip=-0.5mm \tableofcontents}

\section{Introduction}
A considerable global health burden revolves around chronic diseases, such as ischemic heart disease, diabetes and dementia. The spread of these diseases is intricately linked to the age of individuals and the duration of the disease since its onset. Additionally, calendar time plays a critical role in understanding the incidence and mortality rates associated with these conditions.

The incidence rate and mortality rate of chronic diseases depend on both calendar time and the age of affected individuals. Furthermore, the mortality rate for those with chronic diseases is contingent upon the duration since the onset of the disease.

The Illness-Death Model (IDM) serves as a valuable tool for studying chronic diseases and comprises three distinct states: \emph{Healthy},  \emph{Diseased}, and \emph{Dead}. The terms \emph{Healthy} and \emph{Diseased} refer to a specific chronic disease, e.g., heart disease or diabetes. Each subject of the target population is partitioned into exactly one of these three states. For simplicity, we assume that apart from birth into the \emph{Healthy} state, the population is closed, i.e., there is no migration. Incidence and prevalence are key metrics for quantifying new cases in the healthy state over a specific period, with prevalence representing the proportion of individuals in the \emph{Diseased} state at a given time. The transitions between these states are characterized by incidence rate ($i$), mortality rate without disease ($m_0$), and mortality with disease ($m_1$). The corresponding time scales for incidence $i$ and mortality $m_0$ without disease are calendar time ($t$) and age ($a$), while mortality rate $m_1$ with disease may depend on an additional time scale - duration of the disease since its onset ($d$). 
The IDM with the transition rates is shown in Figure \ref{IDM}.

\begin{figure}[h!]
	\centering
	\includegraphics[width=14cm]{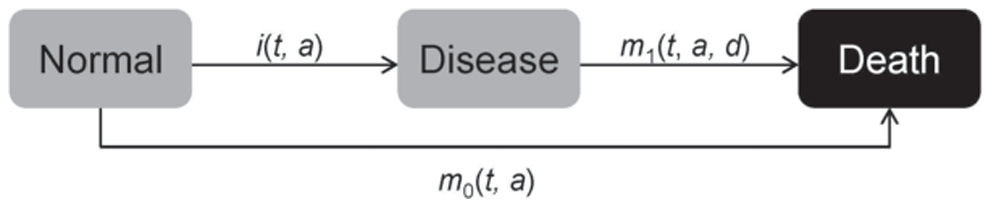}
	\caption {Illness-death model with three states and transition rates $i, m_0$ and $m_1$ depending on time calendar $t$, age $a$ and duration of the disease $d$.}\label{IDM}
\end{figure} 

In some instances, literature may omit the consideration of disease duration $d$. A possible reason for ignoring disease duration is that the stochastic process underlying the IDM looses the Markov property of being memoryless \cite{helwich2008durational}. Age $a$ and calendar time $t$, being important factors for many diseases, can rarely be ignored. Therefore, we consider the most general case, where all three time scales - calendar time, age, and disease duration - will be part of our model.

%
%
%
%
\section{Methods}

First we make analytical considerations and take Keiding's results a bit further. Then we extend Keidings' formula to MLE. Finally, we run a simulation.

\subsection{Keiding's formula for prevalence odds} 
 Let us assume that $S(t,a)$ denotes the absolute number of subjects with age $a$ at time $t$ in the health state and $C(t,a,d)$ the number of cases with age $a$ at time $t$ who are in illness state for duration $d$. The total number of subjects with age $a$ at time $t$ with chronic disease is defined by $C^*(t,a):=\int_{0}^{d}C(t,a,\delta)d\delta$.\\
Notice, in this article we assumed that there is no migration, i.e., no immigration and no emigration. A possible way to extend the theory to the case with migration is given in \cite{brinks2014age}. We also consider that the disease is contracted only after birth. Hence, for all $t$, $C^*(t,0)=0$. And finally we assume that the functions $S$ and $C$ are sufficiently smooth.\\
From the IDM, the following balance equations for the functions $S$ and $C$ will be obtained:
\begin{align}\label{pdes}
	(\partial_t+\partial_a)S(t,a)&=-(m_0(t,a)+i(t,a))S(t,a)\,,\\
	(\partial_t+\partial_a+\partial_d)C(t,a,d)&=-m_1(t,a,d)C(t,a,d)\,,
\end{align}
with the initial conditions:
\begin{align*}
	S(t-a,0)&=S_0(t-a)\,,\\
	C(t,a,0)&=i(t,a)S(t,a)\,.
\end{align*}
 $S(t-a,0)$ indicates the number of healthy newborn people, and $C(t,a,0)$ expresses the number of people who are newly contracted by disease at time $t$ and age $a$. The relation (\ref{pdes}) involves different independent variables and partial derivatives of the unknown function with respect to the independent variables. In mathematics relation (\ref{pdes}) with its initial conditions is called \textit{partial differential equation} or in abbreviated case PDE. \\

\bigskip

The solutions for the PDEs (\ref{pdes}) with above initial conditions are the following equations:
\begin{align}\label{S,C}
	S(t,a)&=S_0(t-a)\exp\Big(-\int_{0}^{a}m_0(t-a+\tau,\tau)+i(t-a+\tau,\tau)d\tau\Big)\,,\\
	\nonumber C(t,a,d)&=C(t-d,a-d,0)\exp\Big(-\int_{0}^{d}m_1(t-d+\tau,a-d+\tau,\tau)d\tau\Big)\,,\\
	\nonumber &=i(t-d,a-d)S(t-d,a-d)\exp\Big(-\int_{0}^{d}m_1(t-d+\tau,a-d+\tau,\tau)d\tau\Big)\,.
\end{align}

The expression for the total number $C^*$ is
\begin{align}\label{Cstar}
	C^*(t,a)&=\int_{0}^{a}i(t-\delta,a-\delta)S(t-\delta,a-\delta)\exp(-\int_{0}^{d}m_1(t-d+\tau,a-d+\tau,\tau)d\tau)d\delta\,.
\end{align}
By using the definition of the age-specific prevalence and applying the equations (\ref{S,C} - \ref{Cstar}), we obtain
\begin{align}
	p(t,a)&=\frac{C^*(t,a)}{S(t,a)+C^*(t,a)}\,.
\end{align}
Brinks and Landwehr \cite{brinks2015new}, have proved the following Lemma:
\begin{lemm}\cite{brinks2015new}
	The total number $C^*$ of diseased persons aged $a\geq 0$ at time $t$, $C^*(t,a)=\int_{0}^{a}C(t,a,\delta) \; d\delta$, is the solution of the initial value problem
	\begin{align*}
		(\partial_t+\partial_a)C^*(t,a)&=-m_1^*(t,a)C^*(t,a)+i(t,a)S(t,a)\,,\\
		C^*(t-a,0)&=0\,,
	\end{align*}
with
	\[
	m_1^*(t,a):=
	\begin{cases}
		\displaystyle\frac{\int_{0}^{a}m_1(t,a,\delta)C(t,a,\delta)d\delta}{\int_{0}^{a}C(t,a,\delta)d\delta} & \text{for}\, C^*(t,a)>0\,,\\
		0 & \text{for}\, C^*(t,a)=0\,.
	\end{cases}
	\]
\end{lemm}
In 1991, Keiding \cite{keiding1991age} has given an expression for the age-specific prevalence-odds $\pi = \frac{C^\star}{S}$:

\begin{thm}
	The prevalence-odds $\pi(t,a)$ of those aged $a\geq 0$ at time $t$ can be calculated by 
 
	\begin{align}\label{keiding}
		\pi(t,a)=\frac{\int_{0}^{a}i(t-a+y, y)\mathcal{M}_{t,a}(y)\exp(-\int_y^a m_1(t-a+\tau,\tau,\tau - y) d\tau) dy}
        {\mathcal{M}_{t,a}(a)}\,,
	\end{align}
	where
	\begin{align*}
		\mathcal{M}_{t,a}(y):=\exp\Big(-\int_{0}^{y}m_0(t-a+\tau,\tau)+i(t-a+\tau,\tau)d\tau\Big)\,.
	\end{align*}
\end{thm}
By knowing the incidence rate $i$ and mortality rates $m_0$ and $m_1$ and by applying the above Theorem, the prevalence-odds $\pi$ can be calculated at time $t$ and age $a.$ Backtransformation via $p = \frac{\pi}{1+\pi}$ yields the prevalence $p.$ Although Keiding's seminal paper \cite{keiding1991age} was well received within the field of biometry and medical statistics, unfortunately the important formula \eqref{keiding} has rarely been used.\\

\bigskip

More than two decades later, Brinks and Landwehr presented a PDE formula, where the prevalence $p$ is the solution of an initial value problem \cite{brinks2015new}: 
\begin{thm}\cite{brinks2015new}
	The age-specific prevalence $p$ is the solution of the initial value problem
	\begin{align}\label{preralph}
	(\partial_t+\partial_a)p=(1-p)(i-p(m_1^*-m_0))\,,
	\end{align}
	with $p(t,0)=0$.
\end{thm}
This Theorem is a generalized work from Brinks et al. \cite{brinks2013deriving} and \cite{brinks2014age}. If we assume that the mortality rate $m_1$ is independent from the duration time $d$, i.e., $m_1=m_1^*$, we obtain
\begin{align*}
	(\partial_t+\partial_a)p=(1-p)(i-p(m_1-m_0))\,.
\end{align*}

Let us assume that the mortality rates without ($m_0$) and with disease ($m_1$) are not known. Then, using the general mortality $m$ of the overall population  $m = p\, m_1^*+(1-p)m_0$ and the relative mortality $R=\nicefrac{m_1^*}{m_0}$, Equation \eqref{preralph} reads
\begin{equation}\label{preralph2}
	(\partial_t+\partial_a)p=(1-p)\Big(i-m\frac{p(R-1)}{p(R-1)+1}\Big)\,.
\end{equation}

The advantages of the PDE \eqref{preralph2} become evident when compared to Keiding's formula \eqref{keiding}. Specifically, if information about $m_0$ and $m_1$ is unavailable, and instead, general mortality $m$ and relative mortality $R$ are provided, the calculation can seamlessly proceed to determine the prevalence $p$ using the associated PDE formula.\\
By comparing equations (\ref{keiding}) and (\ref{preralph}), it can be realized that the PDE formula is sometimes more practicable for estimating the age-specific incidence rate $i$, because Equation \eqref{preralph} can directly be solved for the incidence. Given the prevalence $p$ (with its derivative) and mortality rates $m_1^*$ and $m_0$, we obtain
\begin{equation*}
i=\frac{(\partial_t+\partial_a)p}{1-p}+p(m_1^*-m_0)\,.
\end{equation*}
This formula is useful to approximate the incidence based on the data from two cross-sectional studies. In \cite{brinks2015new}, Algorithm 2.1, gives a detailed explanation about the obtained incidence from two cross-sectional studies.\\
In a study conducted by Brunet and Struchiner \cite{brunet1999non}, they formulated a PDE formula linked to prevalence odds $\pi$
\begin{align*}
	(\partial_t+\partial_a)\pi=(i-(m_1-m_0))\pi+i\,.
\end{align*}

Much like Keiding's formula, the applicability of this PDE formulation is undermined when essential information about $m_0$ and $m_1$ is unavailable. Another notable strength of formula (\ref{preralph}) lies in its ability to deal with migration. Unlike Keiding and Brunet-Struchiner, who do not account for migration, Equations \eqref{preralph} and \eqref{preralph2} can incorporate it seamlessly \cite{brinks2014age}.

In the following subsection, we will extend Keiding's formula and bring it to a simpler form.

\subsection{Prevalence odds and Keiding's formula}

According to the definition, the prevalence odds of the prevalence $p$ at age $a$ and time $t$ is
\begin{align}
	\pi(t,a)=\displaystyle\frac{p(t,a)}{1-p(t,a)}\,.
\end{align}
By the following theorem, a new formula will be given for the prevalence odds $\pi$,
\begin{thm}\label{theorem}
	The prevalence odds $\pi$ at age $a$ and time $t$ can be conceptualized as pseudo-convolution $\pi = i \circledast Y$, which means
	\begin{align}\label{pifinal}
		\pi(t,a)=\int_0^a i(t-\delta,a-\delta) \; Y_{t,a}(\delta)\; d\delta\,,
	\end{align}
	where
\begin{align*}
	Y_{t,a}(\delta):=\exp\Big(-\int_{0}^{\delta}m_1(t-\delta+\tau,a-\delta+\tau, \tau)-(i+m_0)(t-\delta+\tau,a-\delta+\tau)d\tau\Big)\,.
\end{align*}
\end{thm}

Before we give a proof of the Theorem, we provide an interpretation of the pseudo-convolution formuala $\pi = i \circledast Y$: The prevalence-odds $\pi$ at the point $(t, a)$ is the $\delta$-integral of the incidence $i$ at $(t-\delta, a - \delta)$ weighted with $Y_{t,a}(\delta)$. $Y_{t,a}(\delta)$ is positive and for $m_1 - m_0 - i > 0$ monotonously decreasing in $\delta$. Hence, for increasing $\delta$ the incidence $i(t-\delta, a - \delta)$ looses its impact on $\pi(t, a).$

\begin{proof}
From Equation \eqref{keiding}, one can obtain the following formula for the prevalence odds
\begin{align}\label{preodd}
  		\pi(t,a)=\int_{0}^{a}i(t-\delta,a-\delta)\frac{\mathcal{M}_{t,a}(a-\delta)}{\mathcal{M}_{t,a}(a)}\exp(-\int_{0}^{\delta}m_1(t-\delta+\tau,a-\delta+\tau,\tau)d\tau)d\delta\,.
\end{align}
The calculations with details can be found in the Appendix. After simplifying the relation, we have 
\begin{equation*}
\frac{\mathcal{M}_{t,a}(a-\delta)}{\mathcal{M}_{t,a}(a)}=\exp\Big(\int_{a-\delta}^{a}m_0(t-a+\tau,\tau)+i(t-a+\tau,\tau)d\tau\Big)\,.
\end{equation*}

Hence,
\small{\begin{align}\label{pre-pi1}
\nonumber	\pi&(t,a)\\
	&=\int_{0}^{a}i(t-\delta,a-\delta)\exp\Big(\int_{a-\delta}^{a}m_0(t-a+\tau,\tau)+i(t-a+\tau,\tau)d\tau\Big)\exp(-\int_{0}^{\delta}m_1(t-\delta+\tau,a-\delta+\tau,\tau)d\tau)d\delta\,.
\end{align}}

By defining $\gamma:=\tau-a+\delta$, we have $\tau=\gamma+a-\delta$ and $d\tau=d\gamma$, which yields
\begin{align*}
	\exp\Big(\int_{a-\delta}^{a}m_0(t-a+\tau,\tau)+i(t-a+\tau,\tau)d\tau\Big)&=\exp\Big(\int_{0}^{\delta}(m_0+i)(t-a+\gamma+a-\delta,\gamma+a-\delta)d\gamma\Big)\\
	&=\exp\Big(\int_{0}^{\delta}(m_0+i)(t-\delta+\gamma,a-\delta+\gamma)d\gamma\Big)\,.
\end{align*}

Therefore, the formula (\ref{pre-pi1}) will turn to
\begin{align}\label{pi}
	\pi(t,a)
	=\tiny{\int_{0}^{a}i(t-\delta,a-\delta)\exp\Big(\int_{0}^{\delta}(m_0+i)(t-\delta+\gamma,a-\delta+\gamma)d\gamma\Big)\exp(-\int_{0}^{\delta}m_1(t-\delta+\tau,a-\delta+\tau,\tau)d\tau)d\delta}\,.
\end{align}
Then, the main result of this paper will be obtained as 
\begin{align*}
	\nonumber\pi(t,a)=\int_{0}^{a}i(t-\delta,a-\delta)Y_{t,a}(\delta)d\delta\,,
\end{align*}
where 
\begin{align*}
	Y_{t,a}(\delta):=\exp\Big(\int_{0}^{\delta}m_0(t-\delta+\gamma,a-\delta+\gamma)+i(t-\delta+\gamma,a-\delta+\gamma)-m_1(t-\delta+\gamma,a-\delta+\gamma,\gamma)d\gamma\Big)\,.
\end{align*}
\end{proof}

\bigskip

In epidemiology literature, sometimes it is assumed that the incidence is an exponential first order polynomial. The following Corollary is a special case of the above Theorem.
\begin{corollary}\label{corollary}
If the incidence function $i$ in Theorem \ref{theorem} is an exponential first order function
\begin{align*}
	i(t,a):=\exp(k_0+k_1a+k_2t)\,,
\end{align*}
with constants $k_0, k_1, k_2,$ then, the pseudo-convolution \eqref{pifinal} will turn to the following convolution
\begin{align}\label{convolution}
	\pi(t,a)=i^*(t,a)\Big(\text{Exp}*Y_{t,a}\Big)(\delta)\,,
\end{align}
where 
\begin{align*}
i^*(t,a)=\exp(k_0+k_1a-k_1t)\,,\qquad \text {Exp}(t-\delta)=\exp((k_1+k_2)(t-\delta))\,.
\end{align*}
\end{corollary}

\begin{proof}
See Appendix \ref{appencoro}.    
\end{proof}
The purpose is to apply the final definition of the prevalence odds in the simulation $R$-code. 

\section{Maximum Likelihood Estimation (MLE)}

\subsection{General remarks}

In our simulated data, we employ the maximum likelihood estimation (MLE) method as a key tool. MLE is a method utilized for estimating and determining parameter values based on observed data. After obtaining these parameter values, the likelihood function is maximized, revealing the most probable set of parameters.

We use the binomial distribution with following likelihood function
\begin{align*}
\mathcal{L}(\gamma|p)=\binom{n}{c}p^{c}(1-p)^{n-c}\,,
\end{align*}
where $\gamma$ is a vector detailed about later and 
\begin{align}\label{preandodd}
	p(\gamma)=\displaystyle\frac{\pi(\gamma)}{1+\pi(\gamma)}.
\end{align}

In this expression, $c$ and $n$ are the numbers of cases and the total number of people in the study (diagnosed plus not diagnosed), respectively. It is common practice to work with the logarithm ($\log$) of the likelihood function. To maximize the function, the first derivative of the logarithmized likelihood function is set to zero. The primary objective of MLE is to identify parameter values $\hat{\gamma}$ that maximize the likelihood function within an appropriately chosen parameter space $\Gamma$. In addition, by inversion of the Hessian matrix at the maximum (Fisher Information matrix), asymptotic confidence intervals can be estimated. The maximum likelihhood estimator $\hat{\gamma}$ is given by
\begin{align}
\hat{\gamma}:=\arg\max_{\gamma \in \Gamma}\log \mathcal{L}(\gamma|p)\,,
\end{align}
with $\Gamma$ being the parameter space. \\

\subsection{Simulation}
In a study by Brinks et al. \cite{brinks2014lexis}, an algorithm was introduced to simulate chronic diseases and compare the results with analytical calculations. This study also incorporated three essential time scales: calendar time, age, and the duration of the disease since its onset. The input data for their simulation included incidence ($i$), mortality without disease ($m_0$), and mortality with disease ($m_1$).

The simulation assumes that the mortality without disease $m_0(t,a)$, is defined as intensity of a Gompertz distribution
\begin{align*}
	m_0(t,a)=\exp(\xi_1+ \xi_2 \, a + \xi_3 \, t)
\end{align*}
with $\xi:=(-10.7, 0.1, \ln(0.998))$. By changing the variables $(t,a)$ to $(t-\delta +\tau, a-\delta+\tau)$ we have
\begin{equation*}
	m_0(t-\delta+\tau,a-\delta+\tau)=\exp\Big\{\xi_1+\xi_2\, (a-\delta)+\xi_3\,(t-\delta)+\tau \,(\xi_2+\xi_3)\Big\}\,.
\end{equation*}

Integrating the right hand side of the above relation gives
\begin{equation*}
\int_{0}^{\delta}\exp\Big\{\xi_1+\xi_2\, (a-\delta)+\xi_3\,(t-\delta)+\tau \,(\xi_2+\xi_3)\Big\}
=\frac{1}{\xi_2+\xi_3}\Big\{m_0(t,a)-m_0(t-\delta,a-\delta)\Big\}=:M_0(t,a,\delta).
\end{equation*}

The incidence rate in simulation is given by
\begin{align*}
	i(t,a)=\displaystyle\frac{(a-30)_+}{3000}\,,
\end{align*}
where $x_+$ means the positive part of $x$, $x_+:=\max(0,x)$. 
We define $I(t,a,\delta)$ to be
\begin{equation*}
	\frac{1}{3000}\int_{0}^{\delta}(a-\delta+\tau-30)_{+}d\delta=:I(t,a,\delta).
\end{equation*}
In the process of simulation, the mortality with disease $m_1$ is assumed to be factorized into
\begin{align}\label{mortalrate}
	m_1(t,a,d)=m_0(t,a)\,R(d)\,.
\end{align}

Here, $R(d)$  is defined by a quadratic function
	$R(d)=\gamma_1 \, (d-\gamma_2)^2+\gamma_3
$ with $\gamma = (0.04, 5, 1).$

The simulation encompasses 11 distinct age groups, ranging from 40 to 95 years old, with a total participant count of 74,388 individuals and a cross sectional study at $t=100$. The total number of participants and the instances of death are extracted from running Simulation 2 in \cite{brinks2014lexis}. 
\begin{table}[H]
	\caption{The table displays age groups along with the corresponding number $n_k$ of people alive in the study and the number $c_k$ of (prevalent) people in the \emph{Diseased} state.}
	\hspace{-1cm}
	\begin{center}
		\begin{tabular}{cccc}
		Index &	Age group & Nr. of people alive & Nr. of people with disease\\
         $k$ &    & $n_k$ & $c_k$\\
			\hline
		1&	40-44& 9858& 283\\
		2&	45-49& 9786&  501\\
		3&	50-54& 9597& 781\\
		4&	55-59& 9328& 1145\\
		5&	60-64& 8857& 1228\\
		6&	65-69& 8040& 1347\\
		7&	70-74& 6873& 1240\\
		8&	75-79& 5329&  997\\
		9&	80-84& 3706&  679\\
		10&	85-89& 2104&  370\\
		11&	90-94&  910&  164\\
			\hline
        &  & 74388&8735
		\end{tabular}
	\end{center}
\end{table}

Figure \ref{prevPlot}, shows the resulting age-specific prevalence odds at age $a$ and $t=100$. The blue line indicates the analytically calculated prevalence odds. The black circles show the simulated prevalence odds regarding to 11 different age groups:
\begin{figure}[H]
	\centering
	\includegraphics[width=12cm]{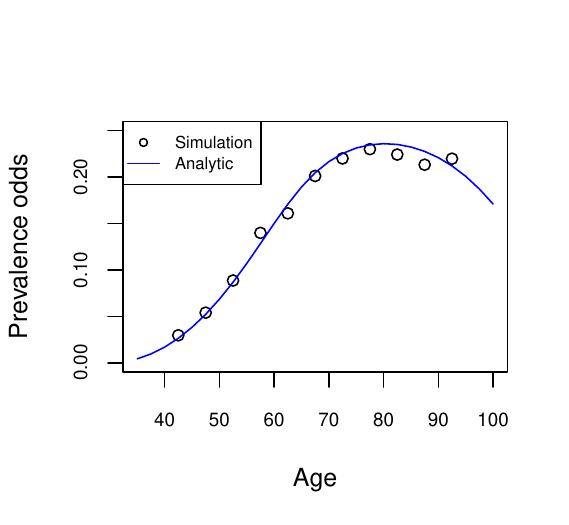}
	\caption {Age-specific prevalence odds $\pi$ for the ages between 30-100 years at year $t=100.$ The black circles are the prevalence odds based on the data in Table 1, the blue line is the numerical calculation using Equation \eqref{pifinal}.}\label{prevPlot}
\end{figure}

Assuming stochastic independence of the numbers in the 11 age groups in Table 1, the formulation of the likelihood function for our specific case reads:
\begin{align}\label{likelihood}
	\mathcal{L}(\gamma)=\prod_{k=1}^{11} \binom{n_k}{c_k} \, p_k \left (\gamma \right )^{c_k} \; \left ( 1-p_k \left (\gamma \right ) \right )^{n_k-c_k}\,,
\end{align}
Wherein, each time, the counts of diagnosed cases and the total number depend on the respective age group. Our goal is to estimate the parameters $\gamma = (\gamma_1, \gamma_2, \gamma_3)$ by maximizing the aforementioned likelihood function. The (asymptotic) confidence bounds for the components of $\gamma$ are obtained by inverting the Hesse-matrix \cite{wood2015core}. The results are given in Table 2.

\begin{table}[H] \label{restab}
	\caption{Input for the simulation and estimates with 95\% confidence intervals for each component of $\gamma$. }
	\hspace{-1cm}
	\begin{center}
		\begin{tabular}{c| c c l}
		   & Input & Point estimate & 95\% confidence interval\\
			\hline
	      $\gamma_1$ & 0.04 & 0.0330 & (-0.0127, 0.0787)\\
        $\gamma_2$ & 5    & 3.06   & (-5.70, 11.8) \\
        $\gamma_3$ & 1    & 1.01   & (0.625, 1.39)
		\end{tabular}
	\end{center}
\end{table}

\section{Conclusion}
In this paper, we examine and extend the prevalence odds formula proposed by Keiding \cite{keiding1991age}. While Keiding's formula is presented in a somewhat complex manner, potentially posing challenges for users in the epidemiology field, we have undertaken the task of simplifying it. Our aim is to make the formula more accessible for interested individuals within the epidemiology community.

To achieve this simplification, we introduce a unique pseudo-convolution formula, $\pi=i\circledast Y_{t,a}$, for the prevalence odds in Theorem \ref{theorem}. Notably, our theorem considers the incidence rate as an arbitrary function, offering flexibility in its application. Additionally, we present Corollary \ref{corollary}, where we specialize our approach by assuming the incidence rate to be a first-order exponential function. In this particular scenario, the pseudo-convolution $\pi$ transforms into a real convolution, $\pi(t,a)=i^*(t,a)(Exp*Y_{t,a})(\delta)$, enhancing the practicality and clarity of the formula.

One notable advantage of employing the prevalence odds $\pi = i \circledast Y_{t,a}$ lies in its ability to capture the potential dependency between mortality with disease $m_1$ and the incidence $i$, extending to potential covariates. This feature opens avenues for future research, encouraging exploration of additional covariates that may influence the relationship. For those interested in delving deeper into this aspect, we recommend referring to the work by Hoyer et al. \cite{hoyer2019risk}. Their study provides valuable insights that can contribute to a more comprehensive understanding of the topic.

A potential drawback of utilizing the pseudo-convolution $\pi=i \circledast Y_{t,a}$ may arise from its dependence on a comprehensive time period. In essence, to compute the prevalence odds, it is imperative to possess sufficient data regarding $m_0$, $m_1$, and $i$ spanning from birth (at $(t-a, 0)$) up to the current time point $(t,a)$. To illustrate, consider an individual aged 90 who is part of the study; in this scenario, information covering the entire 90-year period leading up to the present is required. Rarely, epidemiological data comprise so long periods. Indeed, considering a Partial Differential Equation (PDE) with appropriate initial or boundary conditions may be a more practical and efficient approach in case long-term data are not available. 

After formulating the prevalence odds, the subsequent step involved modeling through simulation. The study encompassed 11 distinct age groups. Utilizing R-code and Maximum Likelihood Estimation (MLE), coupled with the Fisher Information matrix, we successfully derived lower and upper confidence bounds. These bounds were instrumental in estimating the parameters associated with the identified mortality rate $R=\frac{m_1^*}{m_0}$.

There is another advantage of the MLE approach, which is not followed in depth here, but is mentioned for sake of completeness: general properties of maximum likelihood estimators can easily be applied to the situation described in this paper, i.e., estimator for $\gamma$ is unbiased, consistent and reaches the Cramer-Rao lower bound \cite[Chp. 4]{wood2015core}.

Intriguingly, in the simulated population, we only need the numbers $n_k$ and $c_k$ of people alive to make inference about the mortality rate ratio $R$. Thus, no observations of the numbers of deaths from the \emph{Health} or \emph{Diseased} state were necessary. In nuce, we can say, that we can infer information about mortality by looking at the people alive. Furthermore, the simulation approach holds promise for informing strategic interventions and services for individuals grappling with chronic diseases. By gaining insights into the risk factors through this methodology, it becomes possible to formulate effective strategies to enhance the well-being of the affected population.

\section{References}

\bibliographystyle{alpha}

\bibliography{keid}

\begin{appendix}
\section{Appendix}
\subsection{Calculations for Theorem \ref{theorem}}
In this section, we will do the calculations of Theorem \ref{theorem}, step by step. By defining $$M_1(t,a,d):=\int_{0}^{d}m_1(t-d+\tau,a-d+\tau,\tau)d\tau\,,$$ we have
\begin{align}\label{pre}
\nonumber 1&-p(t,a)\\
\nonumber&=1-\frac{\int_{0}^{a}i(t-\delta,a-\delta)\mathcal{M}_{t,a}(a-\delta)\exp(-M_1(t,a,\delta))d\delta}{\mathcal{M}_{t,a}(a)+\int_{0}^{a}i(t-\delta,a-\delta)\mathcal{M}_{t,a}(a-\delta)\exp(-M_1(t,a,\delta))d\delta}\\
\nonumber&\scriptsize{=\frac{\mathcal{M}_{t,a}(a)+\int_{0}^{a}i(t-\delta,a-\delta)\mathcal{M}_{t,a}(a-\delta)\exp(-M_1(t,a,\delta))d\delta-\int_{0}^{a}i(t-\delta,a-\delta)\mathcal{M}_{t,a}(a-\delta)\exp(-M_1(t,a,\delta))d\delta}{\mathcal{M}_{t,a}(a)+\int_{0}^{a}i(t-\delta,a-\delta)\mathcal{M}_{t,a}(a-\delta)\exp(-M_1(t,a,\delta))d\delta}}\\
&=\frac{\mathcal{M}_{t,a}(a)}{\mathcal{M}_{t,a}(a)+\int_{0}^{a}i(t-\delta,a-\delta)\mathcal{M}_{t,a}(a-\delta)\exp(-M_1(t,a,\delta))d\delta}
\end{align}
Now we are ready to calculate $\pi$. By using the relation (\ref{pre}) we get
\begin{align*}
\pi(t,a)&=\displaystyle\frac{p(t,a)}{1-p(t,a)}\\&=\displaystyle\frac{\displaystyle\frac{\int_{0}^{a}i(t-\delta,a-\delta)\mathcal{M}_{t,a}(a-\delta)\exp(-M_1(t,a,\delta))d\delta}{\mathcal{M}_{t,a}(a)+\int_{0}^{a}i(t-\delta,a-\delta)\mathcal{M}_{t,a}(a-\delta)\exp(-M_1(t,a,\delta))d\delta}}{\displaystyle\frac{\mathcal{M}_{t,a}(a)}{\mathcal{M}_{t,a}(a)+\int_{0}^{a}i(t-\delta,a-\delta)\mathcal{M}_{t,a}(a-\delta)\exp(-M_1(t,a,\delta))d\delta}}\\
&=\displaystyle\frac{\int_{0}^{a}i(t-\delta,a-\delta)\mathcal{M}_{t,a}(a-\delta)\exp(-M_1(t,a,\delta))d\delta}{\mathcal{M}_{t,a}(a)}\,.
\end{align*}
Since $\mathcal{M}_{t,a}(a)$ is independent from $\delta$, we may take it under the integral. Therefore,
\begin{align*}
\pi(t,a)&=\int_{0}^{a}i(t-\delta,a-\delta)\frac{\mathcal{M}_{t,a}(a-\delta)}{\mathcal{M}_{t,a}(a)}\exp(-M_1(t,a,\delta))d\delta\,.
\end{align*}

Notice that
\begin{align*}
&\frac{\mathcal{M}_{t,a}(a-\delta)}{\mathcal{M}_{t,a}(a)}\\
&=\frac{\exp\Big(-\int_{0}^{a-\delta}m_0(t-a+\tau,\tau)+i(t-a+\tau,\tau)d\tau\Big)}{\exp\Big(-\int_{0}^{a}m_0(t-a+\tau,\tau)+i(t-a+\tau,\tau)d\tau\Big)}\\
&=\frac{\exp\Big(-\int_{0}^{a-\delta}m_0(t-a+\tau,\tau)+i(t-a+\tau,\tau)d\tau\Big)}{\exp\Big(-\int_{0}^{a-\delta}m_0(t-a+\tau,\tau)+i(t-a+\tau,\tau)d\tau-\int_{a-\delta}^{a}m_0(t-a+\tau,\tau)+i(t-a+\tau,\tau)d\tau\Big)}\\
&=\frac{\exp\Big(-\int_{0}^{a-\delta}m_0(t-a+\tau,\tau)+i(t-a+\tau,\tau)d\tau\Big)}{\exp\Big(-\int_{0}^{a-\delta}m_0(t-a+\tau,\tau)+i(t-a+\tau,\tau)d\tau\Big).\exp\Big(-\int_{a-\delta}^{a}m_0(t-a+\tau,\tau)+i(t-a+\tau,\tau)d\tau\Big)}\\
&=\frac{1}{\exp\Big(-\int_{a-\delta}^{a}m_0(t-a+\tau,\tau)+i(t-a+\tau,\tau)d\tau\Big)}
\end{align*}

Hence, 
\begin{align}\label{pre-pi}
\nonumber	\pi(t,a)&=\int_{0}^{a}i(t-\delta,a-\delta)\frac{1}{\exp\Big(-\int_{a-\delta}^{a}m_0(t-a+\tau,\tau)+i(t-a+\tau,\tau)d\tau\Big)}\exp(-M_1(t,a,\delta))d\delta\\
&=\int_{0}^{a}i(t-\delta,a-\delta)\exp\Big(\int_{a-\delta}^{a}m_0(t-a+\tau,\tau)+i(t-a+\tau,\tau)d\tau\Big)\exp(-M_1(t,a,\delta))d\delta\,.
\end{align}

\subsection{Calculations for Corollary \ref{corollary} }\label{appencoro}
By the definition 
\begin{align*}
	i(t,a):=\exp(k_0+k_1a+k_2t)\,,
\end{align*}
we have,
\begin{align*}i(t-\delta,a-\delta)&:=\exp(k_0+k_1(a-\delta)+k_2(t-\delta))=\exp(k_0+k_1a-k_1\delta+k_2t-k_2\delta)\\
&=\exp(k_0+k_1a+k_2t-(k_1+k_2)\delta)=\exp(k_0+k_1a+k_2t-k_1t+k_1t-(k_1+k_2)\delta)\\
&=\exp((k_0+k_1a-k_1t)+(k_1+k_2)(t-\delta))\\
&=\exp(k_0+k_1a-k_1t).\exp((k_1+k_2)(t-\delta))\\
&=i^*(t,a) \text {Exp}(t-\delta),
\end{align*}
Notice that the function $i^*(t,a)$ is independent from $\delta$.

Regarding to this situation we may consider $i^*(t,a)=i^*_T(t).i^*_A(a)$, where the letters $T$ and $A$, are presenting the two variables time and age, respectively \cite{ades1993modeling}. Hence, the convolution (\ref{convolution}) would be seen like
\begin{align*}
\pi(t,a)=i^*_T(t). i^*_A(a)\Big(\text{Exp}*Y_{t,a}\Big)(\delta)\,.
\end{align*}
Now we consider a case that the incidence function $i$, is any arbitrary function (not necessarily expressed by exponential function). It is possible to assume that the incidence function $i$ for the variables $(t-\delta,a-\delta)$ has the following expression
\begin{align*}
i(t-\delta,a-\delta)=i^*(t,a) J(t-\delta)\,,
\end{align*}
which leads us to 
\begin{align*}
\pi(t,a)=i^*(t,a)\Big(J*Y_{t,a}\Big)(\delta)\,.
\end{align*}
By assuming again that $i^*(t,a)=i^*_T(t).i^*_A(a)$, we will have 
\begin{align*}
\pi(t,a)=i^*_T(t).i^*_A(a)\Big(J*Y_{t,a}\Big)(\delta)\,.
\end{align*}
\end{appendix}
\end{document}